\documentclass[runningheads]{llncs}

\usepackage{todonotes}
\usepackage{quantikz}
\usepackage{physics}
\usepackage{amssymb}
\usepackage{amsmath}
\usepackage{mathtools}
\usepackage{float}
\usepackage[T1]{fontenc}
\usepackage{graphicx}
\usepackage[hidelinks]{hyperref}
\hypersetup{
    colorlinks,
    linkcolor={red!50!black},
    citecolor={blue!50!black}
}
\usepackage{orcidlink}
\usepackage{color}

\usepackage{nicematrix}
\newenvironment{smallpmatrix}
  {\left(\begin{smallmatrix}}
  {\end{smallmatrix}\right)}
\usepackage{wrapfig}
\usepackage{arydshln}
\usepackage{algorithm}
\usepackage[noend]{algpseudocode}
\usepackage{enumitem}
\usepackage{booktabs}
\usepackage[capitalise]{cleveref}
\crefname{section}{Sec.}{Secs.}
\Crefname{section}{Section}{Sections}
\crefname{definition}{Def.}{Defs.}
\Crefname{definition}{Definition}{Definitions}
\crefname{example}{Ex.}{Ex.}
\Crefname{example}{Example}{Examples}
\usepackage{xspace}
\def\paulilim{PauliLim\xspace}

\usepackage{listings}
\usepackage{xcolor}
\definecolor{codegreen}{rgb}{0,0.6,0}
\definecolor{codegray}{rgb}{0.5,0.5,0.5}
\definecolor{codepurple}{rgb}{0.58,0,0.82}
\definecolor{backcolour}{rgb}{0.95,0.95,0.92}

\lstdefinestyle{mystyle}{
    backgroundcolor=\color{backcolour},
    commentstyle=\color{codegreen},
    keywordstyle=\color{magenta},
    numberstyle=\tiny\color{codegray},
    stringstyle=\color{codepurple},
    basicstyle=\ttfamily\footnotesize,
    breakatwhitespace=false,
    breaklines=true,
    captionpos=b,
    keepspaces=true,
    numbers=left,
    numbersep=5pt,
    showspaces=false,
    showstringspaces=false,
    showtabs=false,
    tabsize=2
}
\usepackage{tikz}
\usetikzlibrary{shapes,backgrounds,calc,arrows,automata}
\usetikzlibrary{graphs,quotes}
\usetikzlibrary{arrows.meta}
\usetikzlibrary{decorations.pathreplacing}
\usetikzlibrary{shapes,backgrounds,calc,arrows,positioning}
\usetikzlibrary{decorations.pathreplacing,calligraphy}
\tikzset{main/.style={draw,circle}}
\tikzset{leaf/.style={draw,minimum width=1.2em,minimum height=1.2em}}
\tikzset{e0/.style={draw,->,dotted,>=latex}}
\tikzset{e1/.style={draw,->,>=latex}}

\usepackage{pgf}
\newlength{\pgfcalcparm}
\newlength{\pgfcalcparmm}

\usepackage{arydshln}

\everymath=\expandafter{\the\everymath\displaystyle}

\DeclareRobustCommand{\lnode}[5][]{%
  ~\raisebox{-.mm}{%
    \pgftext{\settowidth{\global\pgfcalcparm}{\scriptsize $\,\,\,#2\,\,\,$}}%
    \pgftext{\settowidth{\global\pgfcalcparmm}{\scriptsize $\,\,\,#4\,\,\,$}}%
  \tikz{%
  \vspace{-1mm}%
    \node[state,inner sep=0pt,minimum size=10pt] (v){\scriptsize $#1$};%
    \node[state,inner sep=0pt,minimum size=10pt,left=\pgfcalcparm of v](v0){\scriptsize $#3$};%
    \draw[dotted, ->] (v) to node[above,pos=.45]{\scriptsize $#2$} (v0);%
    \node[state,inner sep=0pt,minimum size=10pt,right=\pgfcalcparmm of v](v1){\scriptsize $#5$};%
    \draw[->] (v) to node[above,pos=.45]{\scriptsize $#4$} (v1)
  }
  }
}

\DeclareRobustCommand{\ledge}[3][]{%
  \pgftext{\settowidth{\global\pgfcalcparm}{\scriptsize $\,\,\,#2\,\,\,$}}%
  \raisebox{-.8mm}{%
  \tikz{%
    \node[inner sep=0pt] (x){$#1\,\,$};%
    \node[state,inner sep=0pt,minimum size=10pt,right=\pgfcalcparm of x](v){\scriptsize $#3$};%
    \draw[->] (x) to node[above,pos=.5]{\scriptsize $\,#2\,\,$} (v);%
  }%
  }%
}

\newcommand{\transpose}{\intercal}

\newcommand{\low}[1]{\textsf{low}(#1)}
\newcommand{\high}[1]{\textsf{high}(#1)}

\newcommand{\qoldder}{{QolDDer}}
\def\randclif{{\sc RandCliff}}
\def\randcliff{{\sc RandCliff}}
\def\mqtbench{{\sc MQT-Bench}}
\newcommand{\tim}[1]{}
\newcommand{\juul}[1]{}
\newcommand{\sebastiaan}[1]{}
\newcommand{\paragraphsentence}[1]{}

\newcommand\defmath[2]{\newcommand#1{\ensuremath{#2}\xspace}}
\defmath\before{\prec}
\defmath\beforeq{\preccurlyeq}

\defmath\LIM{\textsc{LIM}}
\defmath\glim{G\text -\LIM}
\newcommand\Pauli{\textsc{Pauli}\xspace}
\newcommand\pauli{\Pauli}

\newcommand{\leftnospace}{\mathopen{}\mathclose\bgroup\left}
\newcommand{\rightnospace}{\aftergroup\egroup\right}
\newcommand{\argmin}{\operatornamewithlimits{argmin}}
\newcommand{\unit}{I}
\newcommand{\signless}[1]{\textsf{word}\leftnospace(#1\rightnospace)}
\newcommand{\revbin}[1]{\textsf{revbin}\leftnospace(#1\rightnospace)}

\newcommand{\stab}[1]{\text{Stab}(#1)}

\usepackage[firstpageonly=true,angle=0,vpos=.147\paperheight,hpos=.39\linewidth,vanchor=t,hanchor=l]{draftwatermark}
\SetWatermarkText{%
\raisebox{-4cm}
{\begin{minipage}{\linewidth}\centering\href{https://doi.org/10.5281/zenodo.19763549}{\includegraphics[width=12.5mm]
{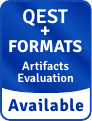}}\hfill\includegraphics[width=12.5mm]
{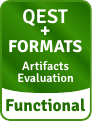}\end{minipage}}%
}

\begin{document}

\title{Faster algorithm for achieving minimal-size quantum decision diagrams}

\author{Juul Sanders
\inst{1}
\orcidlink{0009-0008-5009-1197}
\and
Sebastiaan Brand
\inst{1}
\orcidlink{0000-0002-7666-2794}
\and
Arend-Jan Quist
\inst{2}
\orcidlink{0000-0002-6501-2112}
\and
Tim Coopmans
\inst{1}
\orcidlink{0000-0002-9780-0949}
}

\authorrunning{J. Sanders et al.}

\institute{Delft University of Technology, Delft, The Netherlands
\and
Leiden University, Leiden, The Netherlands\\
\email{\{juul.sanders,s.o.brand,t.j.coopmans\}}\email{@tudelft.nl}\\
\email{a.quist}\email{@liacs.leidenuniv.nl}
}

\maketitle  

\begin{abstract}
\looseness-1
The decision diagram (DD) data structure enables fast linear-algebra calculations by bringing vectors into a normal form and subsequently merging equivalent ones, yielding a minimally-sized DD modulo the equivalence relation.
A fruitful application area is quantum-circuit simulation, where the vectors represent quantum states.
The Local Invertible Map Decision Diagram (LIMDD) type merges LIM-equivalent (typically Pauli-gate equivalent) vectors, 
can efficiently simulate Clifford circuits as well as some high-T-count circuits, 
and has theoretically been proven exponentially faster for simulation than other well-developed data structures, including other common DD variants. 
However, these exponential advantages have not fully materialized yet in existing implementations, for which the normal-form procedure, which is a highly complex algorithm, is either absent or only partially implemented.
We here present a novel normal-form algorithm for Pauli-LIMDDs, achieving a worst-case speedup from $O(n^3)$ to $O(n^2)$ for an $n$-qubit DD node with a single child node while keeping the $O(n^3)$ run time in case of two distinct children nodes.
We implement the algorithm as part of QolDDer, our Pauli-LIMDD simulator for quantum circuits, written from scratch in C/C++.
The implementation realizes the theoretically-proven advantages of Pauli-LIMDDs on Clifford circuits, is significantly faster than the existing LIMDD simulators on such circuits, and on a public quantum-circuit data set often outperforms them by an order of magnitude.
In the future, we envision that our work will enable further application and development of LIMDD variants, not only for quantum design tasks, but also for analysis of linear-algebra-based systems in general.
\end{abstract}

\section{Introduction}
Quantum computing promises computational advantages in domains such as cryptography~\cite{shor1994algorithms}, finance~\cite{orus2019quantum}, optimization~\cite{harwood2021formulating}, and the simulation of quantum physics~\cite{georgescu2014quantum} and chemistry~\cite{cao2019quantum}.
For quantum computers to scale to practically useful tasks, various design challenges need to be solved, such as quantum circuit compilation, performance comparisons of quantum algorithms, and evaluating and mitigating the effects of noise.
An important tool to tackle these challenges is the simulation of quantum circuits on classical computers~\cite{ardeshir2014verification,hong2022equivalence,thanos2023fast}.

Among the algorithms and data structures for quantum circuit simulation are decision diagrams (DD)~\cite{bryant1992symbolic}, which have been proven to be successful in the simulation of quantum circuits and related tasks~\cite{viamontes2004high,miller2006qmdd,samoladas2008improved,zulehner2018advanced,tsai2021bit,wei2022accurate,vinkhuijzen2023limdd,sistla2023symbolic,hong2025limtdd,brand2025qsylvan}.
In quantum computing, states and transformations of states can be represented as pseudo-Boolean functions $f: \{0, 1\}^n \rightarrow \mathbb{C}$, which are exponentially large in the number of qubits.
Decision diagrams are a compact alternative that represents pseudo-Boolean functions as directed graphs. The nodes of these graphs represent sub-functions $f_{\vec{a}} : \vec{x} \mapsto f(\vec{a}, \vec{x})$, for $\vec{a} \in \{0, 1\}^m, m \leq n$.
DDs achieve succinctness due to \emph{canonicity}: by converting each node to a normal form, two nodes representing the same function are \emph{merged}, avoiding storing multiple sub-functions if they are identical.
DDs that offer stronger node merging also exist. For example, in quantum multiple-value decision diagrams (QMDDs~\cite{miller2006qmdd}, i.e., factored edge-valued DDs~\cite{tafertshofer1997factored} with complex edge values), two nodes can be merged if they represent the same sub-function up to some scalar multiple.

Local-Invertible Map Decision Diagrams (LIMDDs) enable even stronger node merging: two nodes are merged if they represent the same sub-function up to some local invertible map (LIM), Kronecker products of invertible $2\times 2$ matrices. Pauli-LIMDDs, where the LIMs consist of matrices from the Pauli group (a set of matrices well known in quantum physics), were proven to be exponentially faster for quantum-circuit simulation than QMDDs and Algebraic Decision Diagrams~\cite{bahar1997algebric,fujita1997multi}, and can, for specific circuits, also outperform common techniques such as matrix product states and Clifford+T simulation~\cite{vinkhuijzen2023limdd,vinkhuijzen2026knowledge}.

Practical implementations of LIMDDs have so far failed to fully realize their theoretical advantages, showcasing the difficulty of their (correct) implementation. The first cause is the complicated $\mathcal{O}(n^3)$ run time merge rule algorithm, solving a group-theoretic task: in practice, a fast implementation is needed to outperform the $\mathcal{O}(1)$ run time merging for QMDDs. 
In addition to this, canonicity, the crucial property for succinct and efficiently updatable decision diagrams, requires careful handling of many (edge) cases.
To make matters worse, bugs in the canonicity implementation are hard to catch, as a non-canonical LIMDD can represent the correct quantum state (and thus pass all typical software tests) while still not being optimal in node number.
Indeed, neither of the two existing proof-of-principle LIMDD implementations is fully canonical. First, MQT-LIMDD~\cite{vinkhuijzen2023efficient}, built as an extension of the QMDD-simulator MQT-DDSIM~\cite{zulehner2018advanced}, does not achieve minimal size in all cases.
Second, LimTDD~\cite{hong2025limtdd}, which implements a generalization of Pauli-LIMDDs, deliberately does not include the canonicity property, voiding any guarantees against (exponential) blow-ups.

In this work, we present a novel algorithm for the Pauli-LIMDD merge rule, implemented as an open-source stand-alone library {\paulilim}~\cite{repopaulilim}. The algorithm has less overhead and achieves a run time improvement from $\mathcal{O}(n^3)$ to $\mathcal{O}(n^2)$ for $n$-qubit nodes with a single child (one of the cases in which DDs achieve a succinctness advantage). We integrate this new algorithm into {\qoldder}, an open-source~\cite{repoqoldder} from-scratch implementation of Pauli-LIMDDs. 
In an empirical evaluation, we compare {\qoldder} against the two previous LIMDD implementations~\cite{vinkhuijzen2023efficient,hong2025limtdd}, as well as against the well-established QMDD implementation MQT-DDSIM~\cite{zulehner2018advanced}. These empirical results corroborate that LIMDDs are able to achieve exponential advantages over QMDDs and underline that LIMDDs require canonicity to realize these advantages.
Our work will enable further application and development of LIMDD variants, not only for quantum design tasks but also for analysis of linear-algebra-based systems in general.

\section{Preliminaries}
\label{sec:prelims}

\subsection{Quantum computing}
\label{sec:prelims-qc}

We briefly introduce quantum computing; see~\cite{nielsen2000quantum} for a full overview.
The state of a single quantum bit (qubit) is a complex-valued 2-vector of unit norm: either $\ket{0} = \begin{bmatrix} 1 & 0 \end{bmatrix}^{\transpose}$ or $\ket{1} = \begin{bmatrix} 0 & 1 \end{bmatrix}^{\transpose}$, or a so-called \emph{superposition} of both, $\begin{bmatrix} \alpha_0 & \alpha_1 \end{bmatrix}^{\transpose} = \alpha_0\ket{0} + \alpha_1 \ket{1}$ where $\alpha_0, \alpha_1$ are complex numbers satisfying $|\alpha_0|^2 + |\alpha_1|^2 = 1$ and where $|.|$ denotes the modulus of a complex number.
The joint state of two qubits,  $\ket{\psi_A}$ and $\ket{\psi_B}$, is obtained from the Kronecker product of their individual states;
$\begin{bmatrix} \alpha_0 & \alpha_1 \end{bmatrix}^\transpose \otimes \begin{bmatrix} \beta_0 & \beta_1 \end{bmatrix}^\transpose = \begin{bmatrix} \alpha_0 \beta_0 & ~\alpha_0 \beta_1 & ~\alpha_1 \beta_0 & ~\alpha_1 \beta_1 \end{bmatrix}^\transpose$.
In general, the state of $n$ qubits can be represented by a vector of unit norm, $\ket{\psi} \in \mathbb{C}^{2^n}$, and can be written in the \emph{Shannon decomposition} $\ket{0} \otimes \ket{\psi_0} + \ket{1} \otimes \ket{\psi_1}$, where $\ket{\psi_0}, \ket{\psi_1}$ are (unnormalized) $(n-1)$-qubit states.

Quantum gates are (reversible) operations that map $n$-qubit states to $n$-qubit states, represented as $2^n \times 2^n$ unitary matrices.
Several common gates are listed below.
Two sets of gates that are particularly relevant to this work are the Pauli gates $\{\unit, X, Y, Z{=}P(\pi)\}$, and the Clifford gates, which are generated from $\{H, P(\tfrac{\pi}{2}), \mathit{CX}\}$ under matrix multiplication and Kronecker products.
\begin{align*}
    {\unit}{=}\begin{bmatrix*}[r] 1 & 0\\ 0 & 1 \end{bmatrix*}~  
    X{=}\begin{bmatrix*}[r] 0 & 1\\ 1 & 0 \end{bmatrix*}~
    Y{=}\begin{bmatrix*}[r] 0 & \text{-}i\\ i & 0 \end{bmatrix*}~ 
    H{=}\frac{1}{\sqrt{2}} \begin{bmatrix*}[r] 1 & 1\\ 1 & \text{-}1 \end{bmatrix*}~
    P(\theta){=}\begin{bmatrix*}[r] 1 & 0\\ 0 & e^{i\theta} \end{bmatrix*}~
    \mathit{CX}{=}\begin{bmatrix} 1&0&0&0\\0&1&0&0\\0&0&0&1\\0&0&1&0\end{bmatrix}
\end{align*}
The quantum state after the application of a gate is obtained through matrix-vector multiplication. 
When a gate is applied to a state with a greater number of qubits, the gate matrix is extended using Kronecker products with $\unit$, so on a 3-qubit state $\ket{\psi}$, applying $X$ on the second qubit can be computed as $\unit\otimes X \otimes \unit \ket{\psi}$.

\looseness=-1
A quantum circuit is a sequence of gates (and measurements) performed on a register of qubits. A straightforward way to simulate quantum circuits is to compute the final state vector by performing matrix-vector multiplication for all gates in the circuit, starting from the initial state. Desired information, such as measurement outcomes, can then be extracted from this final state. The limiting factor of this approach is the exponential ($2^n$) size of $n$-qubit state vectors.

We finish with notions from Pauli stabilizer algebra~\cite{aaronson2004improved,nielsen2000quantum}.
By an $n$-qubit Pauli word $P\in \Pauli_n$, we mean a length-$n$ Kronecker product of the Pauli matrices, e.g. $X\otimes Y\otimes Z$ or $Z\otimes \unit$, abbreviated by omitting `$\otimes$' as  $XYZ$ and $Z\unit$.
We will denote the $2^n \times 2^n$ identity matrix $I_2^{\otimes n}$ as $I_{n}$ and write $\mathbb{I}$ when $n$ is clear from the context.
For this work, a useful representation of a Pauli gate is a binary tuple $(x, z) \in \{0, 1\}^2$, defined as $I\simeq (0, 0), Z\simeq (0, 1), X \simeq (1, 0), Y \simeq (1, 1)$.
We extend this to Pauli words using a \emph{reverse} order than is common~\cite{aaronson2004improved}, defining the binary vector of $P$ as $\revbin{P_1 \otimes P_2 \otimes \dots P_n} = (x_n, x_{n-1}, \dots, x_1, z_{n}, z_{n-1}, \dots, z_1)$.
For example, $\revbin{YZ} = (0, 1, 1, 1)$.

The $n$-qubit Pauli group is the closure of $n$-qubit Pauli words under matrix multiplication.
A stabilizer group of a state $\ket{\phi}$ is the commutative subgroup $\stab{\ket{\phi}} = \{\lambda P \mid \lambda P \ket{\phi} = \ket{\phi}, P\in\pauli_n, \lambda\in\{\pm 1\}\}$ of the Pauli group.
A generating set for a stabilizer group is a set $S$ of signed Pauli words whose closure under multiplication, $\langle S\rangle$, is that stabilizer group.
An example is $\{\unit, -Z\} = \langle -Z\rangle$, with generating set $\{-Z\}$ as $(-Z)^2 = \unit$.
For a signed Pauli word $(-1)^s \cdot P$ with $s\in \{0, 1\}$, we define $\revbin{sP}$ as the concatenation of $\revbin{P}$ and the bit $s$, e.g., $\revbin{-YZ} = (0, 1, 1, 1, 1)$.
Stabilizer generating sets $\{g_1, g_2, \dots, g_k\}$ can be written as binary matrix where the $r$-th row is the binary vector $\revbin{g_r}$.
Such matrices can be brought in row-echelon form (REF) using Gaussian elimination, an elementary linear algebra algorithm, by a slight adjustment of vector addition (see the `rowsum' algorithm in~\cite{aaronson2004improved}).
We will say that a stabilizer generating set is in REF if its binary representation is.

\subsection{Decision diagrams for quantum computing}
\label{sec:prelims-dds}
Decision diagrams (DDs)~\cite{bryant1992symbolic} are a data structure that 
heuristically compresses the representation of (pseudo-)Boolean functions, 
i.e., functions of the form $\{0, 1\}^n \to \mathbb{D}$ for some domain $\mathbb{D}$. 
Quantum states $\ket{\psi}$ and gates $U$ on $n$ qubits can be viewed as pseudo-Boolean functions $\psi : \{0,1\}^n \to \mathbb{C}$ and $U : \{0,1\}^{2n} \to \mathbb{C}$,
where $\psi(\vec x)$ returns the element of the vector $\ket{\psi}$ at index $\vec x$,
which can be encoded in decision diagrams.
A crucial aspect of DDs is that the functions they represent can be manipulated 
\emph{without uncompressing the data structure}~\cite{bryant1992symbolic}.
It is also well known how to do matrix-vector multiplication (i.e., quantum circuit simulation in our context) in the DD representation~\cite{mcmillan1993symbolic}.

In the context of quantum computing, various types of decision diagrams have been used, including
binary decision diagrams (BDDs)~\cite{tsai2021bit,chen2025sliqsim},
multi-terminal BDDs (MTBDDs)~\cite{viamontes2004high,samoladas2008improved},
quantum multiple-valued DDs (QMDDs)~\cite{miller2006qmdd,zulehner2018advanced,brand2025qsylvan},
(weighted) context-free-language ordered BDDs (CFLOBDDs)~\cite{sistla2023symbolic,sistla2024weighted},
and local invertible map DDs (LIMDDs)~\cite{vinkhuijzen2023limdd,vinkhuijzen2023efficient,hong2025limtdd}. 
Because we focus on the effect of canonicity of LIMDDs (\cref{sec:prelims-limdds}), we limit our scope to (canonical and non-canonical) LIMDDs as well as QMDDs, of which LIMDDs are a generalization.

\begin{wrapfigure}[5]{r}{1.9cm}
    \vspace{-9mm}
    \centering
    \begin{tikzpicture}[inner sep=1.3pt, auto, thick]
    	\node[] (0) [] {};
    	\node[main] (1) [node distance=6mm,below of=0] {\phantom{$x$}};
    	\node[main,label=left:{$\psi_0$}] (2) [below left  = 4.5mm and -.2mm of 1] {\phantom{$x$}};
    	\node[main,label=right:{$\psi_1$}] (3) [below right = 4.5mm and -.2mm of 1] {\phantom{$x$}};
    	\draw[e1] (0) to node [] {~$\psi$} (1) ;
    	\draw[e0] (1) to node [above left] {$A$} (2);
    	\draw[e1] (1) to node [] {$B$} (3);
    \end{tikzpicture}
    \vspace*{-7.5mm}
    \caption{}
    \label{fig:three_nodes}
\end{wrapfigure}
Structurally, a DD that encodes a vector $\ket{\psi}$ is a directed acyclic graph whose internal nodes have two outgoing edges. An example of a QMDD or LIMDD node is shown in \cref{fig:three_nodes}.
Such a node represents the Shannon decomposition,
$\ket{\psi} = \ket{0} \otimes A\cdot \ket{\psi_{0}} + \ket{1} \otimes B\cdot \ket{\psi_{1}}$,
where the edge labels $A, B$ are complex numbers for QMDDs, or local invertible maps (LIMs), which are Kronecker products of $2\times 2$ invertible matrices, for LIMDDs.
For a node $v$, we use $\low{v}$ and $\high{v}$ to refer to nodes corresponding to subvectors $\ket{\psi_0}$ and $\ket{\psi_1}$ respectively.
A DD achieves succinctness by recognizing identical or equivalent subvectors and storing only one node for each unique subvector.
This recognition of equivalent subvectors is done efficiently by storing DD nodes in a canonical form (details in \cref{sec:prelims-limdds}). While in general every node has one or two child nodes, there is a single unique terminal node that represents the scalar 1.

QMDDs merge nodes which represent subvectors $\ket{\phi}, \ket{\psi}$ when these satisfy $\ket{\phi} = M \cdot \ket{\psi}$, where $M$ is a complex number. For LIMDDs, this works likewise, but $M$ is a LIM. This broader merging, based on Kronecker products of local ($2 \times 2$) matrices, allows for even stronger compression of the representation. The original work on LIMDDs~\cite{vinkhuijzen2023limdd} mainly focuses on Pauli-LIMDDs, where LIMs are restricted to Pauli-LIMs: a complex number times a Pauli word.

\begin{figure}[t]
    \centering
    \setlength{\tabcolsep}{6pt}
    \begin{tabular}{cccc}
        \scalebox{.98}{\begin{tikzpicture}
\node[] (vec) {
    \begin{minipage}{17mm}
    $\def\arraystretch{1.05}
    \begin{matrix*}[l]
        {000:}\\
        {001:}\\
        {010:}\\
        {011:}\\
        {100:}\\
        {101:}\\ 
        {110:}\\
        {111:}\\
    \end{matrix*}
    \begin{bmatrix*}[r]
        1 \\
        2\\
        3\\
        0\\
        1\\
        2\\ 
        \text{-}3\\
        0\\
    \end{bmatrix*}$
    \vspace{-.5mm}
    \end{minipage}%
};
\end{tikzpicture}} &
        \scalebox{1}{\begin{tikzpicture}[auto, thick,node distance=1.cm,inner sep=1.5pt]
    \newcommand\dist{.5cm}
    \newcommand\sep{.7cm}
    \newcommand\sepp{0.15cm}
    \newcommand\seppp{-.15cm}
    \node[] (0) [] {};
    \node[main] (x0)  [node distance=.65cm,below of=0] {$x_3$};
    \node[main] (x11) [below left  = \dist and \sep of x0] {$x_2$};
    \node[main] (x12) [below right = \dist and \sep of x0] {$x_2$};
    \node[main] (x21) [below left  = \dist and \sepp of x11] {$x_1$};
    \node[main] (x22) [below right = 1.1*\dist and \sepp of x11] {$x_1$};
    \node[main] (x23) [below left  = 1.1*\dist and \sepp of x12] {$x_1$};
    \node[main] (x24) [below right = 1.1*\dist and \sepp of x12] {$x_1$};
    \node[leaf] (l1)  [below left  = .9*\dist and \seppp of x21] {$1$};
    \node[leaf] (l2)  [below right = .9*\dist and \seppp of x21] {$2$};
    \node[leaf] (l3)  [below left  = .9*\dist and \seppp of x22] {$3$};
    \node[leaf] (l4)  [below right = .9*\dist and \seppp of x22] {$0$};
    \node[leaf] (l5)  [below left  = .9*\dist and \seppp of x23] {$1$};
    \node[leaf] (l6)  [below right = .9*\dist and \seppp of x23] {$2$};
    \node[leaf] (l7)  [below left  = .9*\dist and \seppp of x24] {$\text{-}3$};
    \node[leaf] (l8)  [below right = .9*\dist and \seppp of x24] {$0$};
    \draw[e1] (0) to (x0);
    \draw[e0] (x0) to (x11);
    \draw[e1] (x0) to (x12);
    \draw[e0] (x11) to (x21);
    \draw[e1] (x11) to (x22);
    \draw[e0] (x12) to (x23);
    \draw[e1] (x12) to (x24);
    \draw[e0] (x21) to (l1);
    \draw[e1] (x21) to (l2);
    \draw[e0] (x22) to (l3);
    \draw[e1] (x22) to (l4);
    \draw[e0] (x23) to (l5);
    \draw[e1] (x23) to (l6);
    \draw[e0] (x24) to (l7);
    \draw[e1] (x24) to (l8);
\end{tikzpicture}} &
        \scalebox{1}{\begin{tikzpicture}[auto, thick,node distance=1.cm,inner sep=1.5pt]
    \newcommand\dist{.5cm}
    \newcommand\sep{.7cm}
    \newcommand\sepp{0.15cm}
    \newcommand\seppp{-.15cm}
    \node[] (0) [] {};
    \node[main] (x0) [node distance=.65cm,below of=0] {$x_3$};
    \node[main] (x11) [below left = .9*\dist and 1.2*\sepp of x0] {$x_2$};
    \node[main] (x12) [below right = .9*\dist and 1.2*\sepp of x0] {$x_2$};
    \node[main] (x21) [below = \dist of x11] {$x_1$};
    \node[main] (x22) [below = \dist of x12] {$x_1$};
    \node[leaf] (l)   [below = 4.5*\dist of x0] {1};
    \draw[e1] (0) to (x0);
    \draw[e0] (x0) to (x11);
    \draw[e1] (x0) to (x12);
    \draw[e0] (x11) to (x21);
    \draw[e1] (x11) to node [pos=0,right] {~$3$} (x22);
    \draw[e0] (x12) to (x21);
    \draw[e1] (x12) to node [pos=.3,right] {$\text{-}3$} (x22);
    \draw[e0,bend right=30] (x21) to (l);
    \draw[e1,bend left=15] (x21) to node [right,shift=({-.5mm,1.5mm})] {$2$} (l);
    \draw[e0,bend right=15] (x22) to (l);
    \draw[e1,bend left=30] (x22) to node [right,pos=.4] {$0$} (l);
\end{tikzpicture}} &
        \scalebox{1}{\begin{tikzpicture}[auto, thick,node distance=1.cm,inner sep=1.5pt]
    \newcommand\dist{.5cm}
    \newcommand\sep{.7cm}
    \newcommand\sepp{0.15cm}
    \newcommand\seppp{-.15cm}
    \node[] (0) [] {};
    \node[main] (x0) [node distance=.65cm,below of=0] {$x_3$};
    \node[main] (x1) [below = .8*\dist of x0] {$x_2$};
    \node[main] (x21) [below = \dist of x11] {$x_1$};
    \node[main] (x22) [below = \dist of x12] {$x_1$};
    \node[leaf] (l)   [below = 4.5*\dist of x0] {1};
    \draw[e1] (0) to (x0);
    \draw[e1,bend left=30] (x0) to node [] {$1 Z \otimes I$} (x1);
    \draw[e0,bend right=30] (x0) to (x1);
    \draw[e0] (x1) to (x21);
    \draw[e1] (x1) to node [] {$3I$} (x22);
    \draw[e0,bend right=30] (x21) to (l);
    \draw[e1,bend left=15] (x21) to node [right,shift=({-.5mm,1.5mm})] {$2$} (l);
    \draw[e0,bend right=15] (x22) to (l);
    \draw[e1,bend left=30] (x22) to node [right,pos=.4] {$0$} (l);
\end{tikzpicture}} \\
        (a) vector &
        (b) decision tree &
        (c) QMDD &
        (d) LIMDD
    \end{tabular}
    \vspace{-2mm}
    \caption{In this figure, it is shown how a vector of size $2^n$ is represented as a decision tree of depth $n$. A QMDD (c) represents this vector more succinctly by factoring out scalars and storing these on edges. The LIMDD representation (d) compresses the data even further by putting Kronecker products of local ($2 \times 2$) matrices on the edges (d).  \vspace{-\baselineskip}}
    \label{fig:dd-examples}
\end{figure}
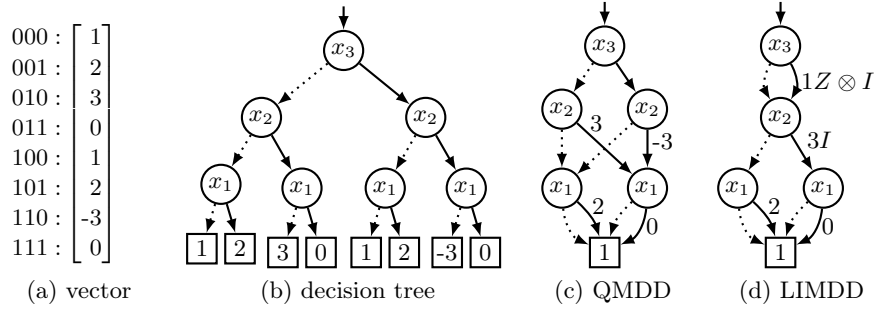

Examples of both data types are found in \cref{fig:dd-examples}.  Here, to simplify the notation, edges with value 1 or $\mathbb{I}$ are left unlabeled. Note that in QMDDs, state vector values can be retrieved by following dashed/solid edges according to the vector index and multiplying all values on the path. Similarly, retrieving values from a LIMDD is done by multiplying the subvectors with the LIMs on the encountered edges.
Observe that in this example, the LIMDD can merge the two $x_2$ nodes because $\begin{bmatrix} 1 & 2 & 3 & 0 \end{bmatrix}^\transpose = (Z \otimes I)\begin{bmatrix} 1 & 2 & \text{-}3 & 0 \end{bmatrix}^\transpose$.

We note that QMDDs can be seen as LIMDDs with LIMs restricted to $\lambda \unit_{2^n}$, with $\lambda \in \mathbb{C}$. Hence, LIMDDs are strictly more succinct than QMDDs.
A more succinct diagram directly benefits quantum-circuit simulation run time as well, as the run time of many DD operations scales with the number of nodes.

Pauli-LIMDDs, similar to the stabilizer formalism~\cite{aaronson2004improved} and unlike QMDDs, can efficiently simulate Clifford circuits (i.e., circuits composed of Clifford gates).
Moreover, it has been shown that there exist classes of circuits where LIMDDs are exponentially faster 
than both QMDDs and the extended stabilizer formalism for the universal Clifford+T gate set, such as W-state preparation~\cite{vinkhuijzen2023limdd}.
LIMDDs have also been shown to be exponentially faster on certain classes of circuits than other simulation methods, such as matrix product states (MPS) ~\cite{vinkhuijzen2023limdd,vinkhuijzen2026knowledge}.

\subsection{Canonicity for Pauli-LIMDDs}
\label{sec:prelims-limdds}
While simulating quantum circuits, the DD algorithms create new diagram nodes.
As mentioned in \cref{sec:prelims-dds}, DDs obtain their succinctness and resulting run time improvements for simulation by recognizing whether an equivalent node is already present in the diagram, so that the new node need not be added.
Although one could implement this by checking the equivalence of the newly created node with each existing node in the diagram, the common (and faster) approach is to bring each node into canonical form.
This reduces the equivalence check with existing nodes in the diagram to an identity check, which is typically implemented using a hash table for $\mathcal{O}(1)$ lookup in practice~\cite{bryant1992symbolic}.

To define a canonical form for LIMDDs, the authors of \cite{vinkhuijzen2023limdd} realized that two LIMDDs in general could represent Pauli-equivalent states, have all identical nodes merged, and still be different, due to the degrees of freedom listed below and visualized in \cref{fig:merge_rules}.

\begin{enumerate}[label=(\alph*)]
    \item Because $0\cdot \vec\psi = 0 \cdot \vec\varphi$, edges with a 0 label can in principle point to any node. This degree of freedom is made canonical by requiring $\low{v} = \high{v}$ whenever one or both has edge label 0.
    \item Because $\ket{0}\otimes\ket{\psi_0} + \ket{1}\otimes\ket{\psi_1} = (X \otimes I^{k}) \left(\ket{1}\otimes\ket{\psi_0} + \ket{0}\otimes\ket{\psi_1}\right)$, there is a degree of freedom regarding the order of $\low{v}$ and $\high{v}$ whenever $X$ is in the LIM group (which it is in Pauli-LIMs). This order is made canonical by requiring $\low{v} \beforeq \high{v}$ according to some total order over all nodes.
    \item Because $(\unit \otimes L) \left(\ket{0}\otimes\ket{\psi_0} + \ket{1}\otimes\ket{\psi_1}\right) = \ket{0}\otimes L\ket{\psi_0} + \ket{1}\otimes L\ket{\psi_1}$, any LIM $L$ can be factored out of the children's edge labels. This is brought into a canonical form by requiring that the label of a low edge pointing to a node representing an $n$-qubit state is $I_n$ (whenever it is not zero).
    \item (The \emph{high determinism} rule) Two nodes can have different high labels and still be Pauli-equivalent (i.e., equivalent up to a Pauli word). For example, given a node $v$ that represents $\ket{0}$,$\lnode{\unit_2}{v}{L}{v}$ represents the (unnormalized) state $\ket{0}\otimes \ket{0} + \ket{1} \otimes \ket{0}$ for both $L=\unit_2$ and $L=Z$.
The high label is brought into a canonical form by choosing the minimum high label according to a defined total order on Pauli-LIMs.
\end{enumerate}

A \emph{reduced} LIMDD is one that does not represent the vector $(0, 0, \dots, 0)^{\transpose}$ and adheres to the above rules as well as the \emph{merge rule}: it has no duplicate nodes.
It is semi-reduced if it satisfies all except high determinism (d).
It was proven in \cite{vinkhuijzen2023limdd} that the above rules eliminate all degrees of freedom: two reduced LIMDD have the same root node if and only if they represent equivalent states.

These rules thus establish a canonical form for nodes: given a newly created node $v$ and edge with label $L$ pointing to it, we want to find a \emph{reduced} node $v'$ and edge label $L'$ such that $v'$ satisfies the rules (a)-(d) above, and $L\ket{v} = L'\ket{v'}$, i.e., they represent the same quantum state.
(If we were performing quantum circuit simulation, we would next check if $v'$ is already present in the LIMDD.)

Given any node $(v,L)$, finding a semi-reduced $(v', L')$, i.e., that satisfies the first three rules (a)-(c), is straightforward and can be done in time $\mathcal{O}(n)$; see \cref{fig:merge_rules} for how to update $(v, L)$.
For Pauli-LIMDDs, for the final rule (d), high determinism, an $\mathcal{O}(n^3)$ algorithm exists \cite{vinkhuijzen2023limdd}, which we improve upon in \cref{sec:hd-algorithm}.

\tikzstyle{oval} = [state, ellipse, minimum size=4mm, inner sep=0.5mm, node distance=1cm]
\tikzset{every picture/.style={->,thick}}

\tikzstyle{leaf}=[draw, rectangle,minimum size=4.mm, inner sep=3pt]
\tikzstyle{var}=[circle,draw=black!70,solid,thick,minimum size=6mm]
\tikzstyle{bdd}=[regular polygon, regular polygon sides=3, draw=black!70,solid,thick,inner sep=0.5mm]
\tikzstyle{n}=[->,loosely dashed,thick]
\tikzstyle{p}=[->,solid,thick]
\tikzstyle{b}=[->,densely dashdotted,ultra thick]
\tikzset{every node/.style={initial text={}, inner sep=2pt, outer sep=0}}
\tikzstyle{lbl}=[draw,fill=white,inner sep=2pt, minimum size=0cm,line width=.5pt]

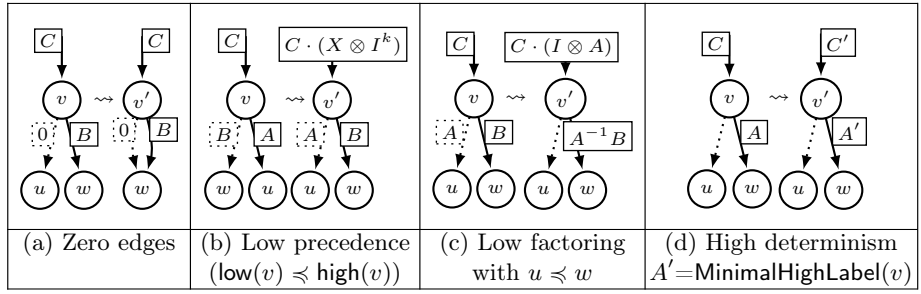
\begin{figure}[t!]
\centering
\footnotesize
\begin{NiceTabular}{|c|c|c|c|}
\hline
&&&\\[-1ex]
        
        
    
    
    
    
    
\hspace{0cm}
    \begin{tikzpicture}[->,>=stealth',shorten >=1pt,auto,node distance=1cm,
        thick, state/.style={circle,draw,minimum size=14pt},font=\scriptsize]
        
    \node[state](r) {$v$};
    \node[state](1a)[below = .7cm of r, xshift=-.29cm]{$u$};
    \node[state](1b)[below = .7cm of r, xshift= .29cm]{$w$};
    
    \node[above = .6cm of r] (x) {};

	\node[right = .1cm of r,yshift=-.0cm,rotate=-0] {$\rightsquigarrow$}; 
    
    \path[]
    (x)  edge[e1]         node[lbl,left,pos=.1] {$C$} (r)
    (r)  edge[e0]     node[left, pos=.3,lbl] {$0$} (1a)
    (r)  edge[e1]     node[right,pos=.3,lbl] {$B$} (1b)
    ;

    \node[state,right = .57cm  of r, inner sep = 0pt](r) {$v'$};
    \node[state](1a)[below = .7cm of r]{$w$};
    
    \node[above = .6cm of r] (x) {};
    
    \path[]
    (x)  edge[e1]         node[lbl,pos=.1] {$C$} (r)
    (r)  edge[e0=15, bend right = 15]     node[left, pos=.2,lbl] {$0$} (1a)
    (r)  edge[e1=15, bend left = 15]     node[right,pos=.2,lbl] {$B$} (1a)
    ;
    \end{tikzpicture}\hspace{0cm}~
&\hspace{-.2cm} 
~ \begin{tikzpicture}[->,>=stealth',shorten >=1pt,auto,node distance=1cm,
        thick, state/.style={circle,draw,minimum size=14pt},font=\scriptsize]
        
    \node[state](r) {$v$};
    \node[state](1a)[below = .7cm of r, xshift=-.29cm]{$w$};
    \node[state](1b)[below = .7cm of r, xshift= .29cm]{$u$};
    
    \node[above = .6cm of r] (x) {};

	\node[right = .2cm of r,yshift=-.0cm,rotate=-0] {$\rightsquigarrow$}; 
    
    \path[]
    (x)  edge[e1]         node[lbl,left,pos=.1] {$C$} (r)
    (r)  edge[e0]     node[left, pos=.3,lbl] {$B$} (1a)
    (r)  edge[e1]     node[right,pos=.3,lbl] {$A$} (1b)
    ;

    \node[state,right = .65cm  of r, inner sep = 0pt](r) {$v'$};
    \node[state](1a)[below = .7cm of r, xshift=-.29cm]{$u$};
    \node[state](1b)[below = .7cm of r, xshift= .29cm]{$w$};
    
    \node[above = .6cm of r] (x) {};
    
    \path[]
    (x)  edge[e1]         node[lbl,pos=.18,xshift=-.7cm] {$C \cdot (X \otimes I^{k})$} (r)
    (r)  edge[e0]     node[left, pos=.3,lbl] {$A$} (1a)
    (r)  edge[e1]     node[right,pos=.3,lbl] {$B$} (1b)
    ;
    \end{tikzpicture}\hspace{0cm}~
&\hspace{0cm}
    \tikz[->,>=stealth',shorten >=1pt,auto,node distance=1.5cm,
        thick, state/.style={circle,draw,minimum size=14pt},font=\scriptsize]{
    \node[state] (1) {$v$};
    \node (x) [above = .6cm of 1] {};
    \node[state] (1a) [below = .7cm of 1, xshift=-.29cm] {$u$};
    \node[state] (1b) [below = .7cm of 1, xshift= .29cm] {$w$};
    
    \path[]
    (x) edge[e1]     node[lbl,left,pos=.2] {$C$} (1)
    (1) edge[e0] node[lbl,above left] {$A$} (1a)
    (1) edge[e1] node[lbl,above right] {$B$} (1b)
    ;
    
	\node[right = .1cm of 1,yshift=-.0cm,rotate=-0] {$\rightsquigarrow$}; 
    
    \node[state,right = .7cm of 1]  (1){$v'$};
    \node (x) [above  = .6cm of 1] {};
    \node[state] (1a) [below = .7cm of 1, xshift=-.29cm] {$u$};
    \node [state](1b) [below = .7cm of 1, xshift= .29cm] {$w$};
    
    \path[]
    (x) edge[e1]       node[draw,lbl,xshift=-.8cm,pos=.3] {$C \cdot (\unit \otimes A)$} (1)
    (1) edge[e0]   node[above left,pos=.7] {} (1a)  %
    (1) edge[e1]   node[draw,above right, pos=.6, lbl,xshift=-.19cm] {$ A^{-1}  B$} (1b)
    ;
    }\hspace{0cm}~
&\hspace{0cm}
    \begin{tikzpicture}[->,>=stealth',shorten >=1pt,auto,node distance=1cm,
        thick, state/.style={circle,draw,minimum size=14pt},font=\scriptsize]
        
    \node[state](r) {$v$};
    \node[state](1a)[below = .7cm of r, xshift= .3cm]{$w$};
    \node[state](1b)[below = .7cm of r, xshift=-.3cm]{$u$};
    
	\node[right = .2cm of r,yshift=-.0cm,rotate=-0] {$\rightsquigarrow$}; 

    \node[above = .6cm of r] (x) {};
    
    \path[]
    (x)  edge[e1]         node[lbl,left,pos=.2] {$C$} (r)
    (r)  edge[e1]     node[lbl, pos=.5] {$A$} (1a)
    (r)  edge[e0]     node[right,pos=.5] {} (1b);

    \node[state,right = .7 of r](r) {$v'$};
    \node[state](1a)[below = .7cm of r, xshift= .3cm]{$w$};
    \node[state](1b)[below = .7cm of r, xshift=-.3cm]{$u$};
    
    \node[above = .6cm of r] (x) {};
    
    \path[]
    (x)  edge[e1]         node[lbl,xshift=0cm,pos=.2] {$C'$} (r)
    (r)  edge[e1]     node[lbl, pos=.5] {$A'$} (1a)
    (r)  edge[e0]     node[right,pos=.5] {} (1b);
    \end{tikzpicture}\hspace{0cm}~
\\[1ex]
\hline
(a) Zero edges & (b) Low precedence & (c) Low factoring  & \hspace{-.1cm}(d) High determinism\hspace{-.1cm} \\
	& ($\low{v} \beforeq \high{v}$) & with $u \beforeq w$  & $A'$=$\textsf{MinimalHighLabel}(v)$\\
\hline
\end{NiceTabular}
\caption{
Normalization rules that bring LIMDD nodes to their reduced form \cite{vinkhuijzen2023limdd}, assuming reduced children $u, w$.
Reproduced and adjusted from \cite[Fig.11]{vinkhuijzen2023limdd} (CC-BY 4.0). \vspace{-\baselineskip}}
\label{fig:merge_rules}
\end{figure}

\subsection{Related work: implementations}
\label{sec:prelims-related-work}

One implementation of these Pauli-LIMDDs is MQT-LIMDD ~\cite{vinkhuijzen2023efficient}, which is an extension of the QMDD implementation MQT-DDSIM~\cite{zulehner2018advanced}.
However, MQT-LIMDD implements canonicity partially, and does not achive minimal size LIMDDs in all cases.
Aside from this implementation, LimTDDs~\cite{hong2025limtdd} have extended LIMDDs to a more expressive group of LIMs, 
specifically all matrices generated under Kronecker and matrix products from $\{X, P(\tfrac{2\pi}{N})\}$ (and a complex scalar), 
where $N \in \mathbb{N}$ can be chosen freely. 
This XP group is a generalization of the one generated by the Pauli matrices, as setting $N=2$ gives $P(\pi) = Z$, $P(\pi)X = iY$, and $XX=I$
which yields Pauli-LIMs.
However, the implementation accompanying~\cite{hong2025limtdd} does, by design, not guarantee canonicity and thus possibly yields LIMDDs that are not minimal in size.

\section{Faster algorithm for High Determinism}
\label{sec:hd-algorithm}

Here, we derive a new algorithm for transforming a semi-reduced Pauli-LIMDD node into one that satisfies the high determinism rule (\cref{sec:prelims-limdds}) with the following advantages over the original algorithm of \cite{vinkhuijzen2023limdd}:
(a) For the case that the node has a single child ($n$ qubits, $m$ stabilizer generators), a run time of $\mathcal{O}(nm)$ instead of $\mathcal{O}(nm^2)$;
(b) For the case of two distinct children, the same worst-case $\mathcal{O}(n^3)$ run time but with less overhead.
Our advantage is gained by adjusting a well-known algorithm for computing intersections and unions of vector spaces.

\subsection{High determinism algorithm: problem statement}
\label{sec:hd-problem-statement}
Given $n$-qubit semi-reduced node $u := \lnode{I}{v_0}{\beta B}{v_1}$ with $\beta \in \mathbb{C}, B \in \pauli_n$, it was shown in \cite[Corollary 15]{vinkhuijzen2023limdd} that the node $u_\text{can} := \lnode{I}{v_0}{\alpha A}{v_1}$ ($\alpha \in \mathbb{C}, A \in \pauli_n$) is a valid choice for the reduced version of $u$, where
\begin{equation}
\label{eq:ucan}
\alpha A = \min_{\substack{s, x \in \{0, 1\},\\ g\in \stab{v_0},\hspace{0.075cm} h \in \stab{v_1}}} \{(-1)^s \cdot \beta^{(-1)^x} g B h \mid x \neq 1 \textnormal{ if } v_0 \neq v_1\}
.
\end{equation}
Indeed, $u$ and $u_\text{can}$ are Pauli-equivalent as~\cite[Theorem 14]{vinkhuijzen2023limdd}
\begin{equation}
\label{eq:xzgmin}
\ket{u_\text{can}} = (X \otimes \beta B)^x \cdot (Z^s \otimes g_\text{min}) \ket{u}
\end{equation}
where $g_\text{min} \in \stab{v_0}$ achieves the minimum in \cref{eq:ucan}.
Canonicity is enforced by the minimization, defined through the following total order~\cite{vinkhuijzen2023limdd}.
In the definition below, for a complex number $r\cdot e^{i\theta}$ in polar form, i.e. $r \in \mathbb{R}_{\geq 0}$ and $\theta \in [0, 2\pi)$ and $P$ a Pauli word, we denote the vector $\vec{re^{i\theta} P}$ as the concatenation of $\revbin{P}$ and the tuple $(r, \theta)$, for example $\vec{3\cdot e^{\pi/2} \cdot YZ} = (0, 1, 1, 1, 3, \pi/2)$.

\begin{definition}[Pauli-LIM total order]
\label{def:vectorization}
For $r, r' \in \mathbb{R}_{\geq 0}, \theta, \theta'\in [0, 2\pi)$ and $P, P'$ Pauli words on the same number of qubits, we say that $r e^{i\theta} P < r' e^{i \theta'} P'$ if $\vec{r e^{i \theta} P} < \vec{r' e^{i \theta'} P'}$ in lexicographic ordering from left to right.
\end{definition}

For example, $\vec{-\tfrac{1}{2}X} = \vec{\tfrac{1}{2} e^{i\pi} X} = (1, 0, \tfrac{1}{2}, \pi) > (0, 1, 2, 0) = \vec{2 e^{0} Z} = \vec{2 Z}$.

\begin{wrapfigure}[9]{r}{2.8cm}
    \centering
    \vspace{-8mm}
    \begin{tikzpicture}[inner sep=1.5pt, auto, thick] 
        \node[draw,circle] (1) {$v_{sem}$};
        \node[draw,circle] (2) [below left=1.1cm and  .5mm of 1] {$v_0$};
        \node[draw,circle] (3) [below right=1.1cm and .5mm of 1] {$v_1$};
        \draw[->,dashed] (1) to node [pos=.45,left] {} (2); 
        \draw[->] (1) to node [pos=.3,right,xshift=-4.5mm,fill=white] {$i\unit Y$} (3);
        \node[] (0) [node distance=9mm,above of=1] {};
        \draw[->] (0) to node [pos=.35,left] {$XYZ$} (1) ;
        \node[draw,circle] (4) [right=.5cm of 1] {$v_{red}$};
        \draw[->,dashed] (4) to node [pos=.25,above] {} (2); 
        \draw[->] (4) to node [pos=.45,right,xshift=-2mm,fill=white] {$iY\unit$} (3);
        \node[] (10) [node distance=9mm,above of=4] {};
        \draw[->] (10) to node [pos=.35,left] {$Y\unit X$} (4) ;
    \end{tikzpicture}
    \vspace*{-7mm}
    \caption{\vspace{-\baselineskip}}
    \label{fig:zassenhaus}
\end{wrapfigure}

\begin{example}
\label{ex:high-det-example}
Consider the semi-reduced node $v_{sem}$ in \cref{fig:zassenhaus} and suppose that $\stab{v_0} = \langle XX, ZZ\rangle = \{\unit\unit, XX, ZZ, -YY\}$ and $\stab{v_1} = \langle XX, YY\rangle = \{\unit\unit, XX, -ZZ, YY\}$.
By trying out all $4\times 4 \times 2 \times 2 = 64$ combinations of $g,  h, s, x$, we find that the minimum \cref{eq:ucan} is achieved by $(g_{\min}, h_{\min}, s, x) = (-YY, \unit\unit, 1, 0)$ at $\alpha A = (-1)^1 \cdot i^{(-1)^0} \left(-Y\unit \right) = i Y \unit$, yielding reduced node $v_{red}$.
We update the incoming edge label from $XYZ$ to $XYZ \cdot (X \otimes iY\unit)^0 \cdot (Z^1 \otimes -YY) = XYZ \cdot -ZYY = Y\unit X$ to keep the same represented quantum state.
\end{example}

The goal of the high determinism algorithm is to evaluate eq.~\eqref{eq:ucan}, i.e. to find $\alpha A$ given $\beta B$ and $\stab{v_0},\hspace{0.1cm} \stab{v_1}$.
In general, there are multiple tuples $(g, h, s, x)$ that achieve the minimum of \cref{eq:ucan}.
However we show in \cref{sec:gbh-proof} that any two such tuples $(g, h, s, x),\hspace{0.1cm} (g', h', s', x')$ satisfy $g\cdot B\cdot h = \pm g'\cdot B \cdot h'$.
Because the Pauli-LIM ordering is lexicographic, i.e., the scalars only matter if the Pauli words are equal, if we found $g,h$ that is off by a factor of $-1$ of the true minimum, we correct for this by choosing $s$ accordingly to get the right factor $(-1)^s$.
This is why we may \emph{greedily} perform the minimization in eq.~\eqref{eq:ucan}: We first find the Pauli word $A$ as $\min_{g \in \stab{v_0}, h \in \stab{v_0}} \signless{gBh}$, where $\signless{L}$ is the Pauli word of Pauli LIM $L$, e.g. $\signless{+XY} = XY$ and $\signless{-iZX} = ZX$.
We give our algorithm for doing so in \cref{sec:hd-alg-min-high-label}.
Next, we find the minimal scalar $\alpha$ by trying the four combinations of $x, s \in \{0, 1\}$.

\subsection{Finding the minimal Pauli word $A$}
\label{sec:hd-alg-min-high-label}
We here assume that we are given $n$-qubit stabilizer generating sets $S_0, S_1$ and an $n$-qubit Pauli word $B$.
We will describe an algorithm to compute
\begin{equation}
\label{eq:minimal-objective}
(g_\text{min}, h_\text{min})
= \argmin_{g\in \langle S_0\rangle, h \in \langle S_1 \rangle} \signless{gBh}
.
\end{equation}
In the context of the high determinism rule, $\langle S_0\rangle$ ($\langle S_1 \rangle$) is the stabilizer group of child node $v_0$ ($v_1$), but our algorithm is agnostic to the interpretation of $S_0$ and $S_1$, except for whether $v_0$ and $v_1$ are identical reduced nodes.

The algorithm from \cite{vinkhuijzen2023limdd} to solve eq.~\eqref{eq:minimal-objective} consists of many individual steps, which include the calculation of intersections of stabilizer groups and cosets of stabilizer groups.
We here sidestep having to calculate these by adjusting one part of the algorithm, \emph{the Zassenhaus procedure}, to directly solve eq.~\eqref{eq:minimal-objective}.
Below, in \cref{sec:stabilizerfindingalgo}, we will also re-use the adjusted algorithm's output to derive a faster and simpler Stabilizer Finding algorithm.

Given two sets of basis vectors $V, W \subset \mathbb{R}^m$ for some positive integer $m$, the Zassenhaus algorithm~\cite{luks1997some} finds bases for the vector spaces $V \cap W$ and $\langle V \cup W \rangle$ (where $\langle.\rangle$ denotes closure under vector addition). 
To do so, it defines a matrix of $2m$ columns, which we call the Zassenhaus matrix, with rows $\{v \circ v \mid v \in V\} \cup \{w\circ 000\dots0 \mid w \in W\}$, where $\circ$ denotes vector concatenation and `$000\dots0$' is the zero vector. 
The algorithm then brings the \emph{left register} of the matrix (the leftmost $m$ columns) into row-echelon form.
To read off a basis for the union $\langle V \cup W \rangle$ after this, one takes all vectors in the left register that are not equal to the zero vector, while a basis for the intersection $V \cap W$ is the set of all (nonzero) vectors in the right register where the left vectors are zero.

We here expand the algorithm to \emph{signed} Pauli words (recall that stabilizers can only have scalar $\pm 1$) to create the Zassenhaus matrix $M$ consisting of \emph{three} registers: $M = \{ \signless{g} \circ g \circ (+\unit_{2^n}) \mid g \in S_0\} \cup \{\signless{h}\circ (+\unit_{2^n}) \circ h \mid h \in S_1\}$.
Thus, the left register contains length-$2n$ binary vectors, the middle only elements of $\langle S_0\rangle$, and the right register only elements of $\langle S_1\rangle$.
We next bring the left register in REF by performing Gaussian elimination on $M$, where addition ($+$) of rows is defined using Pauli-LIM multiplication ($\cdot$) as $(\signless{k}, g, h) + (\signless{k'}, g', h') = (\signless{k \cdot k'}, g' \cdot g, h \cdot h')$.
Because $\langle S_0\rangle$ and $\langle S_1 \rangle$ are stabilizer groups, whose elements commute, this addition is commutative, as Gaussian elimination requires.
As a consequence, the resulting matrix $M'$ has the following properties:
\begin{enumerate}[label=(\roman*)]
\item{
\label{item:first-property}
A basis of minimal size for the binary vector space $\langle \signless{g \cdot h} \mid g \in \langle S_0 \rangle, h \in \langle S_1 \rangle \rangle$ is found as the set of vectors in the left register that are nonzero.
This follows directly from the fact that the left register contains (the Pauli words of) all elements of both $S_0$ and $S_1$
}
\item{
\label{item:second-property}
Each left vector is the product of the middle and right Pauli stabilizers: each row $(k, g, h)$ satisfies $k = \signless{g \cdot h}$.
This is an invariant of the Gaussian elimination algorithm and is hence straightforwardly proven by induction on the algorithm's steps.
}
\end{enumerate}
To find $g_\text{min}$ and $h_\text{min}$ from \cref{eq:xzgmin}, we now perform the standard algorithm for finding the minimal vector between a vector space $V$ and a hyperplane $\{v_0 + v \mid v \in V\}$ for a given vector $v_0$~\cite{vinkhuijzen2023limdd}.
First, append a row $\signless{B} \circ (+\unit_{2^n}) \circ (+\unit_{2^n})$ at the bottom of $M$.
Next, add rows downward to minimize $\signless{B}$ in the left register: iterating a column index $c$ from $1$ to $2n$, if the value of the bottom row $b$ at $c$ contains the value $1$ and there is a row $r$ which contains a leading entry $1$ at column $c$, then replace $b$ by $b+r$.
Finally, the bottom row's left register will contain $\signless{g_{\textnormal{min}} \cdot B \cdot h_{\textnormal{min}}}$ by property~\ref{item:first-property} above, while by property~\ref{item:second-property}, $g_\text{min}$ and $h_\text{min}$ are found in the appended row's middle and right registers, respectively.

\begin{example}
\label{ex:zassenhaus-minimization}
Let $\stab{v_0},\hspace{0.1cm} \stab{v_1},\hspace{0.1cm} B=\unit Y$ be as in Ex.~\ref{ex:high-det-example}. We first calculate $\revbin{B} = \revbin{\unit Y} =  1010, \revbin{XX} = 1100, \revbin{ZZ} = 0011$ and $\revbin{YY} = 1111$.

Now, (a) bringing the left register of the Zassenhaus matrix in row-echelon form, after which we (b) minimize the appended row by adding rows from above depending on the column position of their leading entries, yields
\[
\begin{pNiceArray}{c|c|c}
1100 & +XX & +\unit \unit\\
0011 & +ZZ & +\unit \unit\\
1100 & +\unit\unit & +XX\\
1111 & +\unit\unit & +YY\\
\hdashline
1010 & +\unit\unit & +\unit\unit
\end{pNiceArray}
\stackrel{\textnormal{(a)}}{\rightarrow}
\begin{pNiceArray}{c|c|c}
1100 & +XX & +\unit \unit\\
0011 & +ZZ & +\unit \unit\\
0000 & +XX & +XX\\
0000 & -YY & +YY\\
\hdashline
1010 & +\unit\unit & +\unit\unit
\end{pNiceArray}
\stackrel{\textnormal{(b)}}{\rightarrow}
\begin{pNiceArray}{c|c|c}
1100 & +XX & +\unit \unit\\
0011 & +ZZ & +\unit \unit\\
0000 & +XX & +XX\\
0000 & -YY & +YY\\
\hdashline
0101 & -YY & +\unit\unit 
\end{pNiceArray}
\]
The appended row shows $g_{\textnormal{min}}$=$-YY$ and $h_{\textnormal{min}}$=$ +\unit\unit$, recovering Ex.~\ref{ex:high-det-example}'s result.
Note that one could add the third and fourth row too while arriving at the same minimal string `0101' in the bottom row, resulting in the different values $g_{\textnormal{min}} = +XX$ and $h_{\textnormal{min}} = -ZZ$. However, such different values will lead to the same minimal high label $g_{\textnormal{min}} \cdot B \cdot h_{\textnormal{min}}$ modulo potentially a minus sign, which is accounted for separately (see \cref{sec:hd-problem-statement} and App.~\ref{sec:gbh-proof}).
\end{example}

The run time of the algorithm above is dominated by the Zassenhaus procedure, taking time $\mathcal{O}(n \cdot (|S_0| + |S_1|)^2)$, i.e., at most $\mathcal{O}(n^3)$, identically to the original algorithm in~\cite{vinkhuijzen2023limdd}.
However, in the case that $u$ only has a single child node, i.e. $v_0 = v_1$ and hence $\langle S_0 \rangle = \langle S_1\rangle$, the algorithm for finding $A$ is simpler: note that $\min_{g, h \in \langle S_0\rangle} \signless{gBh} = 
\min_{g, h \in \langle S_0\rangle} \signless{\pm B gh}
=\min_{h \in \langle S_0\rangle} \signless{Bh}$ hence we only need to minimize $B$ by multiplying with elements of $\langle S_0\rangle$, i.e. performing the minimal-vector algorithm only, an $\mathcal{O}(n\cdot |S_0|)$ procedure.

As $|S_0| \leq n$, this improves the original algorithm's $\mathcal{O}(n^3)$ run time to $\mathcal{O}(n^2)$.

\subsection{Finding the stabilizer group of a reduced node}
\label{sec:stabilizerfindingalgo}

So far, we have shown how to reduce a node by minimizing~\cref{eq:ucan}, provided that we have generator sets for the stabilizer groups of the node's children, in row-echelon form (REF).
Here, we give a recursive algorithm to find the latter: find an REF generating set for the stabilizer group of a reduced node.
Compared to the original algorithm~\cite{vinkhuijzen2023limdd}, our algorithm has less overhead as it avoids computing group intersections multiple times by re-using the Zassenhaus result from the high determinism algorithm in \cref{sec:hd-alg-min-high-label}, and has a speedup from $\mathcal{O}(n^3)$ to $\mathcal{O}(n^2)$ for reduced-nodes with a single child.

\begin{example}
\label{ex:stabilizer-finding}
Consider the node $v:=\lnode{I}{v_0}{-i X}{v_0}$ with single child node $v_0$ whose state is $\ket{v_0} = \ket{0}$.
The stabilizer group $\stab{v_0} = \{\unit_2, Z\}$ which is generated as $\langle Z\rangle$.
The goal is to find a generating set for $\stab{v}$ given $v_0$'s generating set $\{Z\}$.
Using the fact that $\ket{v} = \ket{0}\otimes \ket{0} - i\ket{1}\otimes \ket{1}$, a brute-force search over all $2\times 4^2$ signed $2$-qubit Pauli words reveals $\stab{v} = \{\unit\unit, -YX, ZZ, -XY\}$.
\end{example}

The base case for our recursive algorithm is the leaf node, which has a stabilizer group consisting only of the number $1$.
For the inductive case, suppose that we have an $n$-qubit reduced node $\lnode{I}{\phi_0}{\alpha A}{\phi_1}$ where $\alpha \in \mathbb{C}, A \in \pauli_{n-1}$ and $\phi_0, \phi_1$ are reduced nodes.
Denote $S_0$ ($S_1$) for an REF (row-echelon form) generating set of node $\phi_0$ ($\phi_1$).
Our goal is to find all $\lambda P$ such that $\lambda P \ket{\phi} = \ket{\phi}$, where $\lambda \in \{\pm 1\}$ and $P \in \pauli_n$.

We first expand $\lambda P = P_1 \otimes P'$ where $P_1 \in \{\unit, X, Y, Z\}$ (absorbing $\lambda$ into the Pauli-LIM $P'$) and distinguish cases $P_1 \in \{I, Z\}$ and $P_1\in \{X, Y\}$.
For the case $P_1 \in \{I, Z\}$ we write $P_1 = \begin{smallpmatrix}1&0\\0&x\end{smallpmatrix}$ for $x \in \{1, -1\}$.
Because $\ket{\phi} = \ket{0} \otimes \ket{\phi_0} + \ket{1}\otimes \alpha A \ket{\phi_1}$, in this case $P_1 \otimes P' \ket{\phi} = \ket{\phi}$ is equivalent to
\begin{equation}
\label{eq:iz-derivation}
	P' \ket{\phi_0} = \ket{\phi_0} \quad \text{and} \quad
	P' \cdot \alpha A \ket{\phi_1} = x \alpha A \ket{\phi_1}
\end{equation}
If $\alpha = 0$, then the solutions $P'$ are precisely all stabilisers of $\ket{\phi_0}$.
That is, the stabilizers of $\phi$ are $\{P_1 \otimes s \mid P_1 \in \{I, Z\}, s \in S_0\}$; we return $\{Z \otimes I_{n-1}\} \cup \{I \otimes s \mid s \in S_0\}$ as generating set.
If $\alpha \neq 0$, we rewrite the latter equation of \cref{eq:iz-derivation} as $x \sigma(A, P') P' \ket{\phi_1} = \ket{\phi_1}$ using the fact that $x = 1/x$ and where $\sigma(\cdot, \cdot) = 1$ if the arguments commute and $-1$ otherwise.
So, a generating set for $\phi$'s stabilizers is
\begin{equation}
\label{eq:zcase}
\left\{\begin{pmatrix}1&0\\0&\sigma(g, A) \cdot \gamma_g\end{pmatrix} \otimes g
	\mid g \in S'\right\}
\end{equation}
where $S'$ is an REF generating set for the set of stabilisers $g$ of $\ket{\phi_0}$ such that $g \ket{\phi_1} = \gamma_g \cdot \ket{\phi_1}$ for some $\gamma_g \in \{\pm 1\}$.

If $v_0 = v_1$, then $S' = S_0$, otherwise $S$ is found as $S' = \{g \mid (000\dots0, g, \gamma \cdot g) \in M'\}$, where $M'$ is the Zassenhaus matrix of $S_0$ and $S_1$ (see \cref{sec:hd-alg-min-high-label}).

\begin{example}
Let $\stab{v_0},\hspace{0.1cm} \stab{v_1}$ be as in Ex.~\ref{ex:high-det-example}.
We already found their Zassenhaus matrix in Ex.~\ref{ex:zassenhaus-minimization}.
The relevant rows are
\[
\begin{pNiceArray}{c|c|c}
0000 & +XX & +XX\\
0000 & -YY & +YY\\
\end{pNiceArray}
=
\begin{pNiceArray}{c|c|c}
0000 & +XX & \gamma_{+XX} \cdot (+XX)\\
0000 & -YY & \gamma_{-YY} \cdot (+YY)\\
\end{pNiceArray}
\]
from which we read off: $S' = \{+XX, -YY\}$ and $\gamma_{+XX} = +1$ and $\gamma_{-YY} = -1$.

\end{example}

Second, for the case $P_1 \in \{X, Y\}$ we write $P_1 = \begin{smallpmatrix}0&y^*\\y&0\end{smallpmatrix}$ with $y \in \{1, i\}$.
Then $P_1 \otimes P' \ket{\phi} = \ket{\phi}$ is expanded as
\begin{equation}
\label{eq:y-derivation}
	y^*\alpha P' A \ket{\phi_1} = \ket{\phi_0} \quad \text{and} \quad
	y P'\ket{\phi_0} = \alpha A \ket{\phi_1}
\end{equation}
If $\alpha = 0$, then there are no solutions because that would imply that $\ket{\phi_0}=\ket{\phi_1}=0$, which is not a reduced node by definition~\cite{vinkhuijzen2023limdd}.

If $\alpha \neq 0$, then $\ket{\phi_0}$ and $\ket{\phi_1}$ are Pauli-isomorphic.
But by assumption, $\phi_0$ and $\phi_1$ are reduced, hence $\phi_0 = \phi_1$.

Using the fact and multiplying both sides of the latter equation with $y^* P'$, \cref{eq:y-derivation} yields the same equation twice
\begin{equation}
\label{eq:y-derivation-two}
y^* \alpha P' A \ket{\phi_0} = \ket{\phi_0}
\end{equation}
where we used that $y^* = 1/y$ and that $P'$ and $A$ are their own inverse.
Since all Pauli-LIMs which are stabilizers have scalar either $1$ or $-1$~\cite{nielsen2000quantum}, while $y\in \{1, i\}$ and the multiplication of $P'$ and $A$ yields a scalar $\in \{\pm 1\}$, it follows that a solution in this case is only possible if $\alpha = \pm y$, which in particular implies that $\alpha \in \{\pm 1, \pm i\}$.
Provided that that is the case, a solution to \cref{eq:y-derivation-two} is $P' = y/\alpha A$.
That is, the second case yields
\begin{equation}
\label{eq:xcase}
\begin{pmatrix}0&\alpha^* \\ \alpha &0\end{pmatrix} \otimes A
\end{equation}
if $\alpha \in \{\pm 1, \pm i\}$, and no stabilizer generators otherwise.
Since any generating set of a stabilizer group can be brought back to a form where only a single generator has $P_1 \in \{X, Y\}$ as the first Pauli entry~\cite{audenaert2005entanglement}, it suffices to return this one only.

\begin{example}
Let $v, v_0$ and $\alpha=-i, A=X$ be as in \cref{ex:stabilizer-finding}.
We calculate a generating set for $\stab{v}$ using the algorithm described above.
As $\stab{v_0}$ is generated by $Z$ and $\sigma(Z, A) = \sigma(Z, X) = -1$, the first case yields 
$Z\otimes Z$.
For the second case, as $\alpha=-i$, we find $y = i$, so a stabilizer of $v$ is 
$y/\alpha \cdot Y \otimes A = -YX$.
Thus $\stab{v} = \langle ZZ, -YX\rangle$, which expands to $\stab{v} = \{\unit \unit, ZZ, -YX, -XY\}$, recovering the full stabilizer group given in \cref{ex:stabilizer-finding}.

\end{example}

The above algorithm returns a generating set for a node's stabilizer group.
To return an REF generating set, we call Gaussian elimination on the output in case $\alpha \neq 0$.
In the case of distinct children with $\alpha \neq 0$, the algorithm's run time is dominated by the Zassenhaus procedure, taking $\mathcal{O}(n \cdot (|S_0| + |S_1|)^2) $ time.
In the case of a single child with $\alpha \neq 0$, the procedure takes $\mathcal{O}(n\cdot |S_0|)$ time to calculate $\sigma(s, A)$ for all $s \in S_0$.
As $S_0$ is already in REF by assumption and our REF definition has the most significant qubit at the right (see $\revbin{}$ definition in \cref{sec:prelims}), the set in \cref{eq:zcase} is in REF already.
Hence, Gaussian elimination only updates the single row appended in \cref{eq:xcase}, taking $\mathcal{O}(n \cdot |S_0|)$ time.
If $\alpha = 0$, then the output is already in REF, as the last row is $\revbin{Z\otimes \unit_{n-1}} = (0000\dots001)$, i.e., it is majorized by $\revbin{\unit \otimes s}$ for any stabilizer $s \neq \unit_{2^n}$.
As $|S_0|, |S_1| \leq n$, the run times are at most $\mathcal{O}(n^3), O(n^2)$ and $\mathcal{O}(n^2)$.

\section{Empirical evaluation\label{sec:evaluation}}

        We implemented our high determinism algorithm (\cref{sec:hd-algorithm}) as an open-source C library called {\paulilim}~\cite{repopaulilim}, and we have created a from-scratch LIMDD implementation in C/C++ named {\qoldder}~\cite{repoqoldder}. Together, these allow us to investigate the effects of canonicity on the simulation of quantum circuits.
        The implementations are separated to improve modularity and extensibility for future work.
        Below, we highlight two important design decisions of QolDDer.

        As is typical for a DD implementation, to efficiently recognize equivalent nodes, {\qoldder}'s nodes are stored in a hash table after bringing them into a canonical form~\cite{bryant1992symbolic}. We resolve hash collisions through chaining. Each node is assigned a unique ID based on its index in the table. These IDs are then used as canonical representations of the children when hashing a parent node.

        While most of the information that uniquely defines a node (qubit number, children nodes,
        Pauli words of edge labels) is discrete and can be hashed directly; scalars need to be treated separately. 
        Floating-point computations are subject to small rounding errors that can prevent nodes from merging (e.g., \texttt{1/sqrt(2)} yields 0.707...5, while \texttt{cos(pi/4)} yields 0.707...6).
        Most DD implementations that include floating points handle this by setting a threshold $\delta$ and consider two scalars $\alpha$ and $\alpha'$ equivalent when $|\alpha - \alpha'| \leq \delta$~\cite{sanner2005affine,somenzi2015cudd,zulehner2019efficiently}.
        In our case, with $\alpha,\alpha' \in \mathbb{C}$, we consider them
        equivalent when $|\alpha_{\text{real}} - \alpha'_{\text{real}}| \leq \delta$ and $|\alpha_{\text{imag}} - \alpha'_{\text{imag}}| \leq \delta$. 
        We leave $\delta$ configurable, and set the default $\delta = 10^{-13}$; for 64-bit floats, $\delta$ should generally be greater than $10^{-15}$~\cite{brand2025numerical}.
        To find $\delta$-close scalars, we store them in a search tree.
        Additionally, in our implementation, we make a minor change to the LIMDD normalization rules that has been shown to mitigate numerical errors in QMDDs~\cite[Fig.2c]{brand2025qsylvan} while preserving LIMDD canonicity. Details are in App.~\ref{app:max_precedence}.

        The correctness of a LIMDD implementation has two components: representing the correct state vector at all times and representing it using reduced nodes only.
        We rigorously tested {\qoldder} on the first aspect by verifying that {\qoldder}'s output, converted to a state vector, is identical to the output of MQT-DDSIM and LimTDD on the entire {\mqtbench} benchmark set~\cite{quetschlich2023mqtbench} (see \cref{sec:tools-and-benchmarks}) up to 10 qubits, and for the first 10 vector entries up to 20 qubits. 
        We tested the second aspect by verifying {\paulilim}'s correctness against a simple brute-force implementation that solves \cref{eq:ucan} on $10^5$ randomly-generated 4 to 8-qubit instances of high edge labels and stabilizer groups. 
    
    \subsection{Empirical results}
    \label{sec:tools-and-benchmarks}
       
        We investigated the effect of canonicity on simulation run time using two benchmark sets.
        First, we generate a set of random Clifford gates, {\randcliff}, as Pauli-LIMDDs are known to outperform QMDDs on such circuits~\cite{vinkhuijzen2023limdd,vinkhuijzen2026knowledge}. It consists of randomly generated $n$-qubit Clifford circuits with $100n$ gates, ranging from $n=10$ to $30$ qubits, each with $10$ instances. Gates are drawn uniformly from the Clifford generating set $\{H,P(\tfrac{\pi}{2}),\mathit{CX}\}$.
        Next, {\mqtbench}~\cite{quetschlich2023mqtbench} is an existing benchmark set comprised of a broad range of circuit types.

        \looseness=-1
        All benchmarks were run on an AMD Ryzen 7 9800X3D CPU, with 64 GB of available memory. 
        Each run uses a 10-minute timeout and a 16 GB memory limit. Reproducible benchmarks are available online~\cite{zenodobenchmarks}.

        \begin{figure}[h!]
            \centering
            \input{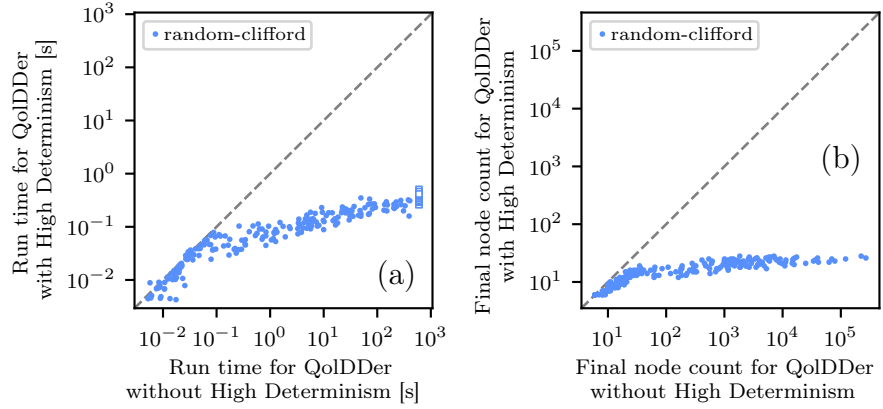}
            \vspace{-2mm}
            \caption{
                 Run times and node counts of canonical and semi-canonical versions of our Pauli-LIMDD simulator {\qoldder} on the {\randclif} benchmark set. Note the log-log scales, on which a monomial function ($x \mapsto x^k$) shows up as a straight line. The dotted grey line visualizes break-even, meaning equal run times or node counts.
            }
            \label{fig:qoldder-no-hd-randcliff}
        \end{figure}

        \looseness=-1
        By turning the use of high determinism on or off, we investigate the effect of simulation on {\randclif} with reduced versus semi-reduced Pauli-LIMDDs (see \cref{sec:prelims-limdds}).
        \cref{fig:qoldder-no-hd-randcliff} shows a run time comparison.
        We observe that full canonicity results in significantly faster run times. 
        On this log-log scale, the data points follow a trend that appears linear, or even sub-linear, indicating a polynomial run time improvement or better.
        To verify that the effect of canonical simulation is due to a reduction in node count, we show the number of nodes in the DD of the output state of the random Clifford circuits in \cref{fig:qoldder-no-hd-randcliff}b.
        Here we see that, indeed, the node count is significantly larger for semi-reduced Pauli-LIMDDs.
                
        \begin{figure}[h!]
            \centering
            \input{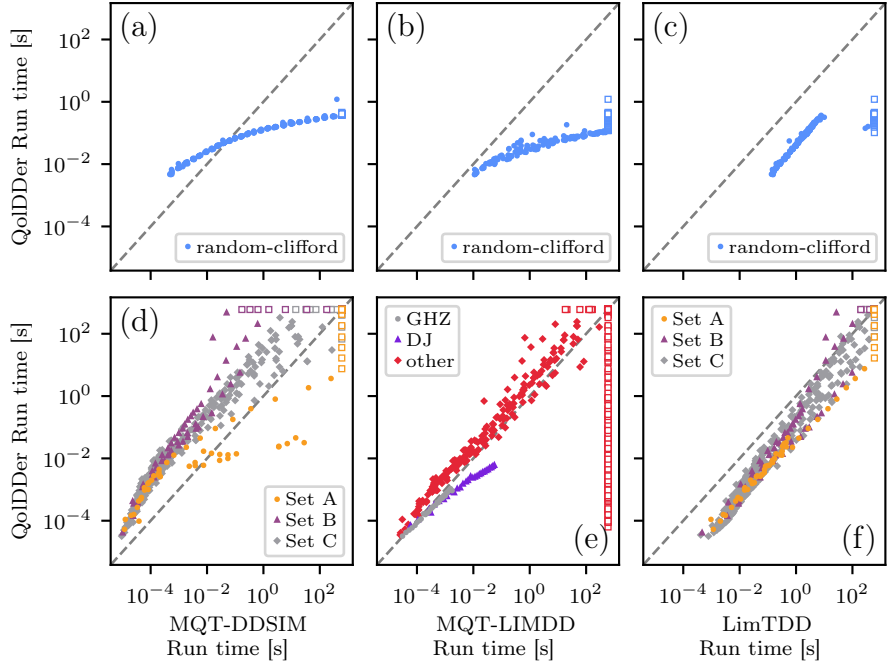}
            \vspace{-2mm}
            \caption{
                Run time comparisons of our implementation ({\qoldder}) to the QMDD implementation (MQT-DDSIM) and the two LIMDD implementations (MQT-LIMDD and LimTDD) on log-log scales. As is shown by the legends, sub-figures (a-c) show the run time comparison for the {\randclif} benchmarks, while sub-figures (d-f) show these for the {\mqtbench} datasets. In the latter, the 27 circuit types have been highlighted according to different categories based on the observed trends.
            }
            \label{fig:results}
         
        \end{figure}

        Next, in \cref{fig:results}, we compare {\qoldder}'s run time against the QMDD simulator MQT-DDSIM and two LIMDD simulators MQT-LIMDD and LimTDD (see also \cref{sec:prelims-related-work}). Based on the observed sublinear scaling in the log-log plot of \cref{fig:results}, we see that {\qoldder} is exponentially faster than QMDDs on {\randcliff}.
        This is as expected, as LIMDDs can represent the output of Clifford circuits in polynomial space, but QMDDs cannot succinctly represent highly entangled states~\cite{vinkhuijzen2026knowledge}, which random Clifford circuits produce~\cite{li2019measurement}.
        We also observe exponential speedup compared to MQT-LIMDD.
        We expect that this is due to the fact that their canonicity is only partially implemented; in theory, our algorithm should only show polynomial improvement (see \cref{sec:hd-algorithm}).
        In \cref{fig:results}c, we observe that {\qoldder} is a constant 25 times faster than LimTDD on some instances, while on larger instances LimTDD times out while {\qoldder} does not.
        This difference is likely due to {\qoldder}'s canonicity, as we observed (not depicted) that LimTDD has much larger node numbers on this benchmark set.

        \looseness=-1
        In \cref{fig:results}d-f, we observe that the advantages of canonicity do not hold on all circuits in {\mqtbench}.
        In \cref{fig:results}d and f, we grouped the benchmarks into sets A = \{Quantum Phase Estimation (QPE) exact, QPE inexact\}, B  = \{Quantum Neural Networks (QNN), Quantum random walks, Grover\}, and C = \{all other\} to highlight the observed trends. In \cref{fig:results}e, only the Greenberger–Horne–Zeilinger state preparation (GHZ) and Deutsch-Josza (DJ) circuits have been highlighted.

        In \cref{fig:results}d, we observe that {\qoldder} significantly outperforms MQT-DDSIM on QPE circuits.
        However, we also observe that the latter is faster on QNN, Quantum random walks, and Grover algorithms (``Set B'').
        We speculate this is because the MQT-Bench benchmark set contains many circuits for which LIMs do not achieve more merging than QMDDs.
        In that case, the LIMs in the DD only add overhead. Nonetheless, we observe roughly a linear scaling on Sets B and C, indicating that {\qoldder} is only slower by a constant factor of $20$.

        In \cref{fig:results}e, we see that {\qoldder}'s performance is similar to MQT-LIMDD on \mqtbench, with two types of outliers. 
        First, on GHZ and DJ circuits, {\qoldder} shows improved scaling over MQT-LIMDD, although due to their sub-1-second run times, we draw no conclusions from these.
        Second, on various circuits, MQT-LIMDD times out or terminates with an error, indicating possible issues in the implementation. In \cref{fig:results}f, we observe a constant-factor speedup of {\qoldder} over LimTDD on QPE (Set A; $30\times$ faster).
        For others, this factor decreases as circuit size increases. 
        We speculate this is due to the larger LIM group that LimTDD has (XP, a supergroup of the Pauli group), enabling more merging, despite not being canonical.

        \begin{figure}[h!]
            \centering
            \input{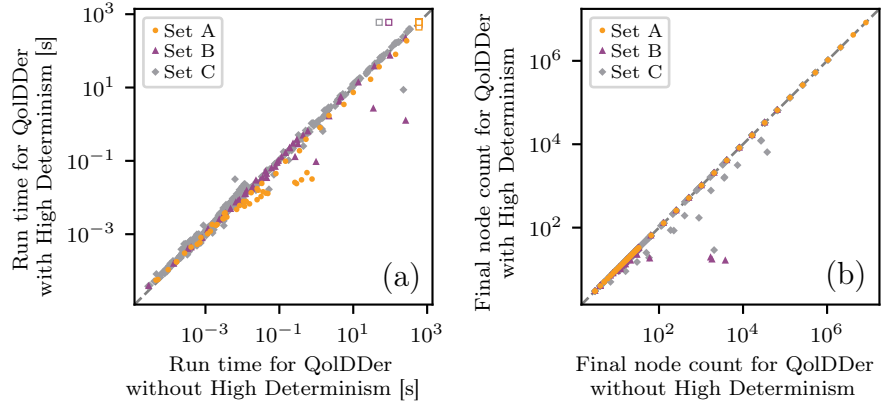}
            \vspace{-2mm}
            \caption{
                Run times and node counts of reduced and semi-reduced versions of our Pauli-LIMDD simulator on the {\mqtbench} benchmark set.
                Note the log-log axes scales.
            }
            \label{fig:qoldder-no-hd-mqtbench}
        \end{figure}

        To investigate the hypothesis that LIMDD canonicity does not speed up run times for {\mqtbench}, we compared semi-reduced and reduced Pauli-LIMDDs in \cref{fig:qoldder-no-hd-mqtbench}. 
        We see little performance difference (\cref{fig:qoldder-no-hd-mqtbench}a), which is explained by a limited gain in succinctness (\cref{fig:qoldder-no-hd-mqtbench}b).

\section{Discussion}

We presented a novel, more efficient algorithm to bring LIMDD nodes into a canonical form, and implemented this algorithm as part of our LIMDD implementation {\qoldder} (both open source~\cite{repoqoldder,repopaulilim}).
In our empirical evaluation, we have compared both run times and node counts for quantum circuit simulation with canonical and non-canonical LIMDDs. Our results corroborate that LIMDDs can achieve exponential advantages over QMDDs. Also, they show that, without canonicity, LIMDDs cannot always realize these advantages. 
On some benchmarks, however, ensuring canonicity is unnecessary overhead.

In the future, we will extend our canonicity algorithm beyond the Pauli group.
Specifically, we expect that our algorithm is extendable to the XP group, which LimTDD uses~\cite{hong2025limtdd}.
We will also apply LIMDDs to a wider range of circuits, such as the universal Clifford+$T$ circuits, for which LIMDDs have scaling guarantees~\cite{quist2026exact}. Moreover, we will look into applications and tasks beyond quantum circuit simulation, such as quantum circuit equivalence checking, a task on which QMDDs have shown success ~\cite{burgholzer2020advanced,brand2025qsylvan}.
In the classical domain, we will investigate the use of LIMDDs for linear-algebra tasks such as stochastic reasoning and Markov-chain analysis, investigating whether relevant tasks possess structure that the strong merge rule of LIMDDs could explicitly make use of. 

\begin{credits}
\subsubsection{\ackname}
SB is supported by the NWO Quantum Technology program (project number NGF.1582.22.035).
AQ acknowledges the support received though the NWO NWA‑ORC program, project Divide \& Quantum (NWA.1389.20.241).
TC is supported by an NWO Veni grant (VI.Veni.232.381). 
\end{credits}

\bibliographystyle{splncs04}
\bibliography{bibs}
\appendix
\section{Proof of High Determinism minimizer form}
\label{sec:gbh-proof}

Fix the number of qubits $n \geq 0$.
Let $G_0, G_1$ be $n$-qubit stabilizer groups and let $B$ be an $n$-qubit Pauli word.
Suppose that both $(g, h, s, x), (g', h', s', x') \in G_0 \times G_1 \times \{0, 1\} \times \{0, 1\}$ achieve the minimum \cref{eq:ucan}.
We here show that $g\cdot B \cdot h = \pm g' \cdot B \cdot h'$.

First, we note that by definition of the total order on PauliLIMs in \cref{def:vectorization}, $\signless{gBh} = \signless{g'Bh'}$.
Hence $gBh = \alpha \cdot g'Bh'$ for some $\alpha \in \mathbb{C}$.
Now we use the fact that stabilizers are their own inverse (because Pauli words square to identity, and the scalars of stabilizers are $\pm 1$), we reorder factors to find that
\[
(g' g) B (h h') = \alpha B
\]
hence
\[
(g' g) B (h h') B = \alpha \unit_{2^n}
.
\]
As Pauli words either commute or anticommute, we find
\[
\pm (g' g) B B (h h') = \alpha \unit_{2^n}
\]
so
\begin{equation}
\label{eq:pmgprimeg}
\pm (g'g) (h h') = \alpha \unit_{2^n} \text{.}
\end{equation}
As $g'\cdot g$ and $h\cdot h'$ are stabilizers, their scalars are $\pm 1$.
Moreover, following \cref{eq:pmgprimeg} we see that $g'g$ and $hh'$ multiply to a multiple of $\unit_{2^n}$, which is only possible for PauliLIMs if $\signless{g' \cdot g} = \signless{h \cdot h'}$.
Hence, we may write $g'\cdot g = \pm W$ and $h \cdot h' = \pm W$ for some Pauli word $W$.
Plugging back into \cref{eq:pmgprimeg}, we find
\[
\pm (\pm W) \cdot (\pm W) = \alpha \unit_{2^n}
.
\]
Hence $\alpha = \pm 1$, which finishes our proof.

\section{Scalar normalization}
\label{app:max_precedence}
Given a node with low edge label $A$ and a high edge label $B$, the low-factoring rule (see \cref{fig:merge_rules}) normalizes the edge labels by multiplying the root edge label with $A$, setting the high edge label to $A^{-1} B$, and leaving the low edge an identity label.
This means that, for $\alpha$ and $\beta$ being the scalar components of $A$ and $B$ respectively, the normalized scalars will be $\alpha' = \alpha/\alpha = 1$ and $\beta' = \beta/\alpha$. However, for QMDDs, it has been shown that the issues that arise from floating-point imprecision can be mitigated by using a different normalization strategy~\cite{brand2025qsylvan}. Specifically, instead of dividing $\alpha$ and $\beta$ always by $\alpha$, dividing by $\max(\alpha, \beta)$ yields fewer numerical errors.

In LIMDDs, this can be implemented by either modifying the low factoring rule (\cref{fig:merge_rules}c) or by modifying the low precedence rule (\cref{fig:merge_rules}b).
We opt for the latter as it is simpler in implementation.
The modified low precedence rule, which we shall call the \emph{max precedence rule}, is given in \cref{def:max-precedence}. In \cref{fig:merge_rules}c, the low factoring rule has ``$u \beforeq w$'' in the caption. This condition would be replaced by ``max precedence holds''.

\begin{definition}[Max precedence]
        \label{def:max-precedence}
    A LIMDD node $\lnode{\unit_n}{u}{\beta \cdot B}{w}$, where $\beta \in \mathbb{C}$ and $B$ a Pauli word, satisfies \emph{max precedence} if at least one of the following three conditions holds:
    \begin{align*}
        \begin{cases}
            u = w\\
            |\beta| < 1\\
            u \before w \textnormal{ and } |\beta| = 1
        \end{cases}
    \end{align*}
    Here, $\beforeq$ denotes a total order over all nodes.
\end{definition}

We now give an alternative definition of a LIMDD to be reduced, where the \emph{Low precedence} rule (Fig.~\ref{fig:merge_rules}b) is replaced by \emph{Max precedence}.

\begin{definition}
\label{def:max-reduced}
    We say that a LIMDD node is \emph{max-reduced} if it satisfies the following normalization rules: \emph{merge rule}, \emph{Zero edges}, \emph{Low factoring}, \emph{High determinism} (\cref{sec:prelims-limdds}) as well as \emph{Max precedence}.
\end{definition}

What remains is to show that max-reduced LIMDD is canonical, which is the equivalent of the proof that reduced LIMDD is canonical in \cite[Theorem 9]{vinkhuijzen2023limdd}.

\begin{theorem}[Max-reduced Pauli-LIMDDs are canonical]
    For each $n$-qubit quantum state $\ket{\phi}$, there exists a unique \emph{max-reduced} Pauli-LIMDD $L$ with root node $v_L$ such that $\ket{v_L}$ is Pauli-equivalent to $\ket{\phi}$.
\end{theorem}
\begin{proof}
Similar to the proof of \cite[Theorem 9]{vinkhuijzen2023limdd}, our proof consists of two parts: first, to show that there exists a max-reduced LIMDD that represents $\ket{\phi}$ up to Pauli-LIMs. Next, uniqueness: to show that there is only one such max-reduced LIMDD.

\textbf{Existence.}
First, by algorithmically translating a reduced LIMDD into a max-reduced LIMDD, existence follows from the existence proof of reduced LIMDD in \cite[Theorem 9]{vinkhuijzen2023limdd}.
The algorithm is straightforward: given a reduced LIMDD $L_{red}$, consider recursively each node $\ledge{C}{v}$ with $\lnode{\alpha\cdot A}{u}{\beta \cdot B}{w}$ starting at the bottom layer, where $C$ is a Pauli-LIM, $\alpha,\beta \in \mathbb{C}$ and $A,B$ a Pauli word, and ensure that it will satisfy \emph{Max reducedness} (\cref{def:max-reduced}) as follows.
First, to apply low factoring, replace node $\ledge{C}{v}$ by $\ledge{C'}{v}$ and replace $v$ by $\lnode{\unit}{u}{\beta'B'}{w}$, where $C'=\alpha C (I\otimes A)$, $\beta'=\alpha^{-1}\beta$ and $B'=A^{-1}B$.
Second, to apply max precedence, if $\beta' \neq 0$ and $u\neq w$, but either $|\beta'|=1 \wedge w \before v$ or $|\beta'| > 1$ is the case, then replace $\ledge{C'}{v}$ by $\ledge{C'\cdot (X\otimes \beta B')}{v}$ and replace $v$ by $\lnode{\unit}{w}{\tfrac{1}{\beta'} \cdot B'}{u}$. After that, the high determinism algorithm is called on the new node $v$.
Then ensure that the \emph{Merge rule} holds in the entire LIMDD by performing the algorithm in \cite[Sec.~4.1]{vinkhuijzen2023limdd}.
Next, repeat the same procedure on the LIMDD layer above, and repeat this procedure until all layers of the LIMDD have been updated.
It is straightforward to check that the resulting LIMDD is max-reduced.

\textbf{Uniqueness.}
The uniqueness proof runs by induction on the number of qubits $n$.
The base case $n=1$ is given in \cite[Theorem 9]{vinkhuijzen2023limdd}.
The induction case starts identically to the uniqueness proof of \cite[Theorem 9]{vinkhuijzen2023limdd}: 
let $L$ and $M$ be max-reduced LIMDDs with root nodes $v_L, v_M$ and $\ket{v_L} \simeq \ket{\phi} \simeq \ket{v_M}$ with the following structure,
    	\begin{align}
    	    \lnode[v_L]{\unit}{v_L^0}{\beta_L B_L}{v_L^1} & & \lnode[v_M]{\unit}{v_M^0}{\beta_M B_M}{v_M^1}
    	\end{align}
where $\beta_L, \beta_M \in \mathbb{C}$ and $B_L, B_M$ are Pauli words.
Then there exists a Pauli-LIM $\lambda P_{top}\otimes P\ne 0$ that maps $\ket{v_L}$ to $\ket{v_M}$ where $\lambda \in \mathbb{C}$, $P_{top}$ a single-qubit Pauli matrix and $P$ an $(n-1)$-qubit Pauli word.
We set out to prove that $v_L = v_M$, i.e., they are represented by the same node.
Expanding the semantics of LIMDDs, we obtain:
        \begin{align}
        \label{eq:canonicity-equation-1}
       \lambda P_{top} \otimes P (\ket{0} \otimes \ket{v_L^0} + \ket{1} \otimes \beta_L B_L \ket{v_L^1})
        =
        \ket{0} \otimes \ket{v_M^0} + \ket{1} \otimes \beta_M B_M \ket{v_M^1}
        .
        \end{align}    
The proof of \cite[Theorem 9]{vinkhuijzen2023limdd} distinguishes the two cases $P_{top} \in \{\unit, Z\}$ and $P_{top} \in \{X,  Y\}$, with subcases $\beta_L, \beta_M \neq 0$ or at least one of $\beta_L, \beta_M$ equals $0$.
As in the proof of \cite[Theorem 9]{vinkhuijzen2023limdd}, \emph{Low precedence} is only invoked in the subcase $P_{top} \in \{X,  Y\}$ with $\beta_L, \beta_M \neq 0$; we only need to consider that case here.
For the remaining cases, the uniqueness proof in \cite[Theorem 9]{vinkhuijzen2023limdd} works for max-reduced too.

Considering the remaining case, we write $P_{top} = \left(\begin{matrix}0 & z^* \\ z & 0\end{matrix}\right)$ for $z \in \{1, i\}$, so that \cref{eq:canonicity-equation-1} gives:
\begin{equation}
\label{eq:max-precedence-proof}
            \lambda z P \ket{v_L^0} = \beta_M B_M \ket{v_M^1}
            \qquad\textnormal{and}\qquad
            \lambda z^*P \beta_L B_L \ket{v_L^1} = \ket{v_M^0}
            .
\end{equation}
As $\beta_L =0 $ or $\beta_M = 0$ was excluded by assumption, hence it follows from \cref{eq:max-precedence-proof} by the induction hypothesis that $v^0_L = v^1_M$ and $v^1_L = v^0_M$.
We now consider each of the three cases in the definition of \emph{Max precedence} (\cref{def:max-precedence}) individually.
The first case in \cref{def:max-precedence} (``$u=w$'') leads to $v_L$ and $v_M$ having the same child and only differing in high label, while we assumed they are Pauli-equivalent; by \emph{High determinism}, it follows that $v_L = v_M$, as desired.
We now show that the remaining two cases in \cref{def:max-precedence} cannot occur.
To do so, we first reorder terms of \cref{eq:max-precedence-proof} and our previously derived $v^0_L = v^1_M$ and $v^1_L = v^0_M$ show that $\alpha \lambda z^* \beta_L P B_L$ is a stabilizer of $\ket{v^0_L}$, and similarly that $\lambda z \beta_M^{-1} B_M P$ is a stabilizer of $\ket{v^1_L}$.
As stabilizers have a scalar of magnitude 1~\cite{nielsen2000quantum}, we find that $|\lambda z^* \beta_L| = 1$ and $|\lambda z \beta^{-1}| = 1$, resulting in $|\lambda| = |\beta_L| = 1/|\beta_M|$ as $|z| = |z^*| = 1$ (recall $z \in \{1, i\}$).
Hence $|\beta_M| \cdot |\beta_L| = 1$, so at least one of $|\beta_M| \geq 1$ or $|\beta_L| \geq 1$.
Hence, out of the remaining two cases of Max precedence in \cref{def:max-precedence}, it follows that at least one of $|\beta_M| = 1$ or $|\beta_L| =1$, i.e., the third case of \cref{def:max-precedence}.
It follows from $|\beta_M| \cdot |\beta_L| = 1$ that $|\beta_M| = |\beta_L| = 1$ hence both nodes $v_L$ and $v_M$ satisfy the third case of \cref{def:max-precedence}.
The ``$u\before w$'' clause in that case leads to $v^0_L \before v^1_L$ and $v^1_M \before v^0_M$, which contradicts our previously derived $v^0_L = v^1_L$ and $v^0_M = v^1_M$.
\end{proof}

\end{document}